\documentclass[11pt]{article}

\usepackage[cp1250]{inputenc}
\usepackage{multirow}
\usepackage{amsthm}
\usepackage{amsmath,amsfonts,amssymb,graphics,verbatim,epsfig
}

\newcommand{\cF}{\mathcal{F}}

\newcommand{\cK}{\mathcal{K}}

\newcommand{\cP}{\mathcal{P}}

\newcommand{\IE}{\mathbb{E}}

\newcommand{\IF}{\mathbb{F}}

\newcommand{\IP}{\mathbb{P}}
\newcommand{\IQ}{\mathbb{Q}}
\newcommand{\IR}{\mathbb{R}}

\newcommand{\be}{\begin{eqnarray*}}
\newcommand{\ee}{\end{eqnarray*}}
\newcommand{\ben}{\begin{eqnarray}}
\newcommand{\een}{\end{eqnarray}}

\newtheorem{theorem}{Theorem}[section]
\newtheorem{lemma}[theorem]{Lemma}

\theoremstyle{definition}

\newtheorem{problem}{Problem}[section]

\theoremstyle{remark}

\numberwithin{equation}{section}
\newcommand{\esssup}{\mathop{\mathrm{ess\,sup}}}

\def\ee{\varepsilon}

\begin{document}

\title{Hedging of equity-linked with maximal success factor}

\author{Przemys\l aw Klusik
\footnote{University of Wroc\l aw, pl.\ Grunwaldzki 2/4, 50-384
Wroc\l aw, Poland, E-mail: przemyslaw.klusik@math.uni.wroc.pl}
}
\maketitle

\begin{abstract}
We consider an equity-linked contract whose payoff depends on the
lifetime of policy holder and the stock price. We assume the limited capital for hedging and we provide with the best
strategy for an insurance company in the meaning of so called succes factor
$\IE^\IP\left[{\mathbf 1}_{\{V_T \geq D)}+{\mathbf 1}_{\{V_T <
D\}}\frac{V_T}{D}\right ]$, where $V_T$ denotes the end value of strategy and $D$ is the payoff of the contract.
The work is a genaralisation of the work of  F\"{o}llmer and Schied \cite{FS2004} and  Klusik and Palmowski \cite{KluPal}, but it considers much more general "incompletness" of the market, among others midterm nonmarket information signals and infitite nonmarket scenarios.
\medskip


\medskip

\end{abstract}
\section{Introduction}
In this paper we consider an incomplete financial market. On such market the seller of the contract can superhedge it. This is very expensive, but  worst-case-scenario resistant stretegy. Very good example here is the case of unit-linked insurance products, where the payoff, a function of price of some financial index, is paid only if the insured is alive at some specified date. The safest strategy is to be prepared for the case every insured survives, but it costs. The competition forces the insurance companies  to assume some risk, as it is very likely (especially in large cohorts) that some insureds die before the maturity. 

One way of controlling the risk of implemented strategy is so called quantile hedging. Classically we assume some allowed capital and implement the strategy maximizing the probability of satysfying all the claims $E\left[{\mathbf 1}_{\{V_T \geq D\}}\right]$, where $D$ denotes the contingent claim and $V_T$ denotes the final value of hedging portfolio. Equivalently we can fix a probability of successful hedging and look for the cheapest strategy. However this approach can be crititisized because it leaves the unsuccessful scenarios completely out of control. That is why the expected success ratio criterion $E\left[{\mathbf 1}_{\{V_T \geq D\}}+{\mathbf 1}_{\{V_T< D\}}\frac{V_T}{D}\right]$ is also being used.

Foellmer \& Leukert \cite{FL99} investigate the general semimartingale setting. Authors point out the optimal strategy for a complete market with maximal  $E\left[{\mathbf 1}_{\{V_T \geq D\}}\right]$. They also show (just) the existence  of such strategy  maximizing expected successful ratio in incomplete market. The proofs are based on various versions of Neymann-Pearson lemma.

Spivak \& Cvitanic \cite{spivak} study a complete market framework of assets modelled with Ito processes. They constructed a strategy with maximal $E\left[{\mathbf 1}_{\{V_T \geq D\}}\right]$ duality approach to utility maximization. The also implement this technique for market with partial observations. Finally they consider the case where the drift of the wealth process is a nonlinear (concave) function of the investment strategy of the agent.

Sekine \cite{sekine} considers a defaultable securities in an incomplete market, where security-holder can default at some random time and receives a payoff modelled by martingale process. Author shows strategy maximizing the probability of successful hedge.

Klusik, Palmowski, Zwierz \cite{KluPalZwie}  solve the problem of the quantile hedging from the point of view of a better informed agent acting on the market. The additional knowledge of the agent is modelled by a filtration initially enlarged by some random variable. 

The approaches above concentrated on complete market frameworks, which couldn't be used for equity- linked contracts. The payoff of equity-linked policies is function of two random
factors: the price of stock or financial index (hence the term equity-linked),
and some insurance-type (i.e nonmarket) event in the life of the owner of the
contract (death, retirement, survival to a certain date, etc.). As
such, the payoff depends on both financial and insurance risk
elements, which have to be priced so that the resulting premium is
fair to both the seller and the buyer of the contract.
There are very few methods providing an appropriate
risk-management in connection of such a contracts which exploit
some imperfect hedging forms, as the mortality risk makes the market
incomplete. 

The maximum success factor was proposed by Klusik \& Palmowski \cite{KluPal}. They consider the equity-linked product where the insurance event can take a finite numer of states and is independent on financial asset modelled with geometric Brownian motion. They construct optimal strategy for both: maximal probability and maximal expected success ratio. In their framework the knowledge about the insurance event is not revealed before the maturity.

 The equity-linked policies have been studied since the
1970s (see Brennan and Schwartz \cite{bs2}, Boyle and
Schwartz \cite{bs}, and Delbaen \cite{d}). Some later works were  Bacinello \cite{bo}, Aase and Persson \cite{ap},\
Ekern and Persson \cite{e}, Boyle and Hardy \cite{bh}, Bacinello \cite{b}, Moeller \cite{m7}, Melnikov \cite{m1} and \cite{m2}, Melnikov and Romanyuk \cite{m8}. 

In this note we state a general problem of optimizing success  factor $E\left[{\mathbf 1}_{\{V_T \geq D)}+{\mathbf 1}_{\{V_T< D\}}\frac{V_T}{D}\right]$ in a incomplete market, as in  Klusik \& Palmowski \cite{KluPal}, but we allow very general flow of information outside the market. We find an optimal strategy using a geometric approach.

This paper is organized as follows. The Section \ref{main_problem}
introduces a model of financial market and the structure of an
insurance product we consider. We also state and give solution of
both problems of hedging. Iin Section \ref{sec:proofs} we give the
proofs of the main result. To make this section readable we moved one technical proof (proof of Lemma \ref{supremum}) outside, to Appendix.

\section{Mathematical model and investment problem} \label{main_problem}
Consider a discounted price process $X=(X_t)_{t \in [0,T]}$ being a semimartingale on a probability space $(\Omega, \cF, \IP)$ with filtration $\IF^X=(\cF^X_t)_{t \in [0,T]}$. We will assume that there is unique equivalent martingale measure on $\cF^X_T$ denoted by $\IR$.

Consider a filtration $\IF=(\cF_t)_{t \in [0,T]}$ and a fixed sequence $0=t_0<t_1<t_2<\cdots<t_n=T$. We will assume that for $t \in (t_i,t_{i+1})$ holds following equality $\cF_t=\cF_t^X\vee \cF_{t_i}$. The interpretation is following: the knowledge modelled by $\IF$ could be augmented by information outside the market only at times $t_0,t_1,\ldots,t_n$. We will assume that $\cF=\cF_T$.

The augmentation of filtration here could be interpreted as the information signal about non-market variables important to the value of contract. An example here could be the "life" part of information about the equity-linked contract. 

With the slight abuse of notation we extend measure $\IR$ to $\sigma$-field $\cF$:
\begin{equation}
\IR(A):=\int_A\frac{d\IR}{d\IP}\Big{|}_{\cF_T^X}d\IP \nonumber
\end{equation}
for $A \in \cF$.

We will consider the contingent claim $D$ being
an $\cF_T$-measurabe  nonnegative random variable and the
replicating investment strategies, which are expressed in terms of
the integrals with respect to $X$. That is, we will deal with the
self-financing admissible trading strategies $(V_0,\xi)$ where  $V_0$ is constant and $\xi$ is $\IF$-predictable process on
$[0,T]$ for which the value process
$$V_{t}=V_{0}+\int_{0}^{t}\xi_u\;{\rm d}X_u,\qquad t\in[0,T],$$
is well defined and generates non-negative wealth:
$$V_t\geq 0,\qquad \IP-{\rm a.s.}$$
for all $t\in [0,T]$.

Denote the set of all equivalent martingale measures by $\cP$. It means that process $X$ follows a martingale in respect to all measures from $\cP$. 
%
\begin{problem}\label{P2} Fix an initial capital $\tilde{V}_0$. Among all admissible
strategies satisfying $V_0 \leq \tilde{V}_0$ find one that
maximizes {\it expected success ratio}:
\begin{equation}
\IE^\IP\left[{\mathbf 1}_{\{V_T \geq D)}+{\mathbf 1}_{\{V_T <
D\}}\frac{V_T}{D}\right ].
\end{equation}
\end{problem}
Define 
\begin{equation}
f(k)=\esssup_{s \in L^1(\Omega,\cF^X_T)}\left\{s:\IE^{\IR}\left[\frac{d\IP}{d\IR}\frac{1_{\{s\leq D\}}}{D}\Big{|}\cF^X_T\right]\geq k\right\}.
\end{equation}
Assume that constant $k$ is defined by:
\begin{equation}\label{defk}
\IE^{\IR}\left[f(k)\right]=\tilde{V}_0
\end{equation}
\begin{theorem} \label{th_P2}
Solution of Problem \ref{P2} is a super replicating strategy of $\min(D,f(k))$. The maximal success ratio is equal to $\IE^{\IP}\left[\min\left(1,\frac{f(k)}{D}\right)\right]$.
\end{theorem}

\section{Solution of Problem \ref{P2}} \label{sec:proofs}
Let us begin with auxilary problem:
\begin{problem} \label {P2'''}

 Find a $\cF^X_T$-measurable random variable $\Gamma$ such that maximizes $\IE^{\IP}\left[\min\left(1,\frac{\Gamma}{D}\right)\right]$ subject to condition $\IE^{\IR}\Gamma \leq \tilde{V_0}$.
 \end{problem}
\begin{lemma}\label{lemmaGamma}
A solution of Problem \ref{P2'''} is given by $\Gamma=f(k)$ where $k$ is defined by \ref{defk}.
\end{lemma}
\begin{proof}
Consider a function
\begin{equation}
F(s):=\IE^{\IR}\left[\frac{d\IP}{d\IR}\min\left(1,\frac{s}{D}\right)\Big{|}\cF^X_T\right]
\end{equation}
it is almost everywhere increasing and concave in respect of $s$, so
\begin{equation}
F(s)\leq F(f(k))+k(s-f(k))
\end{equation}
Put $s=\Gamma$ and take expectation:
\begin{equation}
\IE^{\IR}[F(\Gamma)]\leq \IE^{\IR}[F(f(k))]+k\left(\IE^{\IR}[\Gamma]-\IE^{\IR}[f(k)]\right)
\end{equation}
Finally as $\IE^{\IR}[\Gamma]\leq \tilde{V}_0=\IE^{\IR}[f(k)]$ we get $\IE^{\IR}[F(\Gamma)]\leq \IE^{\IR}[F(f(k))]$, i.e.
\begin{equation}
\IE^{\IP}\left[\min\left(1,\frac{\Gamma}{D}\right)\right]\leq \IE^{\IP}\left[\min\left(1,\frac{f(k)}{D}\right)\right]
\end{equation}
\end{proof}

For random variable $M$ let us introduce the set $\cK^M$ of $\cF_T^X$ -measurable random variables $K$ such that $\IR(M\geq K|\cF_T^X)>0$ almost surely. 
We will use a convention $\overline{M}=\esssup\cK^M$.

Now we solve another auxilary problem:
\begin{problem} \label {P2''}Find a $\cF$-measurable random variable $\phi$ such that $0\leq\phi\leq 1$ and that maximizes $\IE^{\IP}(\phi)$ subject to condition $\IE^{\IR}[\overline{\phi D}]\leq \tilde{V_0}$.
\end{problem}
\begin{lemma} \label{lem_jednamiara}
A solution of Problem \ref{P2''} is  $\phi^*:=\min\left(1,\frac{f(k)}{D}\right)$.
\end{lemma}
\begin{proof}
First, check that $\phi^*$ is in domain of Problem \ref{P2''}:
\begin{equation}
\IE^{\IR}[\overline{\phi^* D}]=\IE^{\IR}\left[\overline{\min\left(1,\frac{f(k)}{D}\right)D}\right]\leq \IE^{\IR} \left[f(k)\right] = \tilde{V_0}.\nonumber
\end{equation}
Second, check that $\phi^*$ maximizes $\IE^{\IP}(\phi)$:

Assume that $\tilde{\phi}$ is some other solution. As $\IE^{\IR}[\tilde{\phi}D]\leq\IE^{\IR}[\overline{\tilde{\phi}D}]\leq\tilde{V}_0$ from Lemma \ref{lemmaGamma} we get: $\IE^{\IP}\phi^*=\IE^{\IP}\left[\min\left(1,\frac{f(k)}{D}\right)\right]\geq \IE^{\IP}\left[\min\left(1,\frac{\overline{\tilde{\phi} D}}{D}\right)\right]$. Further $\IE^{\IP}\left[\min\left(1,\frac{\overline{\tilde{\phi} D}}{D}\right)\right] \geq \IE^{\IP}\left[\min\left(1,\frac{\tilde{\phi} D}{D}\right)\right] \geq  \IE^{\IP}\left[\min\left(1,\tilde{\phi}\right)\right] \geq  \IE^{\IP} \tilde{\phi}$. 
\end{proof}
\begin{problem} \label {P2'} Find a $\cF$-measurable random variable $\phi$ such that $0\leq\phi\leq 1$ and that maximizes $\IE^{\IP}(\phi)$ subject to condition $\IE^{\IQ}[\phi D]\leq \tilde{V_0}$ for all $\IQ \in \cP$ .
\end{problem}
\begin{lemma} \label{lemma_P2'}
A solution of Problem \ref{P2'} is $\phi=\min\left(1,\frac{f(k)}{D}\right)$.
\end{lemma}
\begin{proof}
The proof is imediate from Lemmas \ref{supremum} and \ref{lem_jednamiara}.
\end{proof}

\begin{proof}[Proof of the Theorem \ref{th_P2}]
Take any admissible strategy $(V_0,\xi)$ from the domain of Problem \ref{P2}, i.e such that $V_0\leq \tilde{V_0}$.

For all $\IQ \in \cP$ we have:
\begin{eqnarray}
\tilde{V}_0 &\geq& V_0\nonumber\\
&\geq&\IE^{\mathbb{Q}}V_T\nonumber\\
&\geq&\IE^{\mathbb{Q}}\left[D{\mathbf 1}_{\{V_T\geq D\}}+V_T {\mathbf 1}_{\{V_T < D\}}\right]\nonumber\\
&=&\IE^{\mathbb{Q}}\left[D\phi^{(V_0,\xi)}\right].\label{obl}\nonumber
\end{eqnarray}
where  $\phi^{(V_0,\xi)}={\mathbf 1}_{\{V_T \geq D)}+{\mathbf 1}_{\{V_T < D\}}\frac{V_T}{D}$. Notice that from Lemma \ref{lemma_P2'} we get $\IE^\IP[\phi^{(V_0,\xi)}]\leq \IE^\IP\left[\min\left(1,\frac{f(k)}{D}\right)\right]$, as $\phi^{(V_0,\xi)}$ is in domain of Problem \ref{P2'}. We show that we can choose strategy $(V_0,\xi)$ so that the value $\IE^\IP\left[\min\left(1,\frac{f(k)}{D}\right)\right]$ is attained.

Take  the strategy $(V_0^*,\xi^*)$ being the super replicating strategy of contingent claim $\min\left(D,{f(k)}\right)$. We have $V_0^*=\sup_{\IQ \in \cP}\IE\left[\min\left(D,{f(k)}\right)\right] \leq \tilde{V}_0$ (from Lemma \ref{lemma_P2'}).  We show that for this strategy the value $\IE^\IP\left[\min\left(1,\frac{f(k)}{D}\right)\right]$ is attained:
\begin{eqnarray}
\IE^{\IP}\left[\min\left(1,\frac{f(k)}{D}\right)\right]&\leq& \IE^{\IP}\left[1_{\{V_T^*\geq D\}}+\min\left(1,\frac{f(k)}{D}\right)1_{\{V_T^*< D\}} \right]\nonumber\\
&=&\IE^{\IP}\left[1_{\{V_T^*\geq D\}}+\frac{\min\left(D,{f(k)}\right)}{D}1_{\{V_T^*< D\}} \right]\nonumber\\
&\leq&\IE^{\IP}\left[1_{\{V_T^*\geq D\}}+\frac{V_T^*}{D}1_{\{V_T^*< D\}} \right]\nonumber\\
&=&\IE^{\IP}\left[\phi^{(V^*_0,\xi^*)}\right]
\end{eqnarray}
So $(V_0^*,\xi^*)$ is a solution of Problem \ref{P2}.
\end{proof}
\section{Appendix}
\begin{lemma} \label{supremum} Let $M\in L(\Omega,\cF)$. Then
\begin{equation}
\sup_{\IQ \in \cP} E^{\IQ}M=\IE^{\IR}\overline{M}.
\end{equation}
\end{lemma}
\begin{proof}
Fix a random variable $M$. For each $\IQ \in \cP$ we have an inequality $E^{\IQ}M\leq E^{\IQ}\overline{M}=E^{\IR}\overline{M}$, so $\sup_{\IQ \in \cP}E^{\IQ}M\leq \IE^{\IR}\overline{M}$ and we need to show just that supremum is attainable.

For any $\alpha \in [0,1]$ and $K \in \cK^M$ define a measure $\IR^{\alpha,K}$
\begin{eqnarray}
\frac{d\IR^{\alpha,K}}{d\IR}=\prod_{i=1}^n A^{\alpha,K}_i
\end{eqnarray}
, where
\begin{eqnarray}
A_i^{\alpha,K}= \alpha+\frac{\IR{(M\geq K|\cF_{t_i})}(1-\alpha)}{\IR(M\geq K|\cF^X_{t_i}\vee \cF_{t_{i-1}} )}
\end{eqnarray}

Note that for any  $s<t_i$ we have
\begin{eqnarray}
\IE^{\IR}[A_i^{\alpha,K}|\cF^X_{T}\vee \cF_s]&=&\IE^{\IR}\left[\alpha+\frac{\IR{(M\geq K|\cF_{t_i})}(1-\alpha)}{\IR(M\geq K|\cF^X_{t_i}\vee \cF_{t_{i-1}} )}\Big{|}\cF^X_{T}\vee \cF_s\right]\nonumber\\
&=&\IE^{\IR}\left[\alpha+\frac{\IR{(M\geq K|\cF_{t_i})}(1-\alpha)}{\IR(M\geq K|\cF^X_{t_i}\vee \cF_{t_{i-1}} )}\Big{|}\cF^X_{t_i}\vee \cF_{t_{i-1}}\right]\nonumber\\
&=&1\nonumber
\end{eqnarray}
and
\begin{eqnarray}
\IE^{\IR}[A_i^{\alpha,K}|\cF_s]&=&\IE^{\IR}\left[\alpha+\frac{\IR{(M\geq K|\cF_{t_i})}(1-\alpha)}{\IR(M\geq K|\cF^X_{t_i}\vee \cF_{t_{i-1}} )}\Big{|}\cF_s\right]\nonumber\\
&=&\IE^{\IR
}\left[\IE^{\IR}\left[\alpha+\frac{\IR{(M\geq K|\cF_{t_i})}(1-\alpha)}{\IR(M\geq K|\cF^X_{t_i}\vee \cF_{t_{i-1}} )}\Big{|}\cF_{t_i}^X \vee \cF_{t_{i-1}} \right]\Big{|}\cF_s\right]\nonumber\\
&=&\IE^{\IR}[1|\cF_s]=1\nonumber
\end{eqnarray}
So for $t \in (t_j,t_{j+1})$

\begin{eqnarray}
\IE^{\IR}\left[\frac{d\IR^{\alpha,K}}{d\IR}\Big{|}\cF_t\right]&=&\IE^{\IR}\left[\prod_{i=1}^n A^{\alpha,K}_i\Big{|}\cF_t\right]=\IE^{\IR}\left[\IE^{\IR}\left[\prod_{i=1}^n A^{\alpha,K}_i\Big{|}\cF_{t_{n-1}}\right]\Big{|}\cF_t\right]\nonumber\\
&=&\IE^{\IR}\left[\prod_{i=1}^{n-1} A^{\alpha,K}_i\IE^{\IR}\left[ A^{\alpha,K}_n\Big{|}\cF_{t_{n-1}}\right]\Big{|}\cF_t\right]=\IE^{\IR}\left[\prod_{i=1}^{n-1} A^{\alpha,K}_i\Big{|}\cF_t\right]\nonumber\\
&=&\IE^{\IR}\left[\prod_{i=1}^{n-2} A^{\alpha,K}_i\Big{|}\cF_t\right]=\ldots=\IE^{\IR}\left[\prod_{i=1}^{j+1} A^{\alpha,K}_i\Big{|}\cF_t\right]\nonumber\\
&=&\prod_{i=1}^{j} A^{\alpha,K}_i\IE^{\IR}\left[A^{\alpha,K}_j\Big{|}\cF_t\right]=\prod_{i=1}^{j} A^{\alpha,K}_i
\end{eqnarray}
In the same way we get
\begin{eqnarray}
\IE^{\IR}\left[\frac{d\IR^{\alpha,K}}{d\IR}\Big{|}\cF_T^X \vee \cF_t\right]=\prod_{i=1}^{j} A^{\alpha,K}_i\nonumber,
\end{eqnarray}
i.e for every $t \in [0,T]$ we have
\begin{eqnarray}
\IE^{\IR}\left[\frac{d\IR^{\alpha,K}}{d\IR}\Big{|}\cF_T^X \vee \cF_t\right]=\IE^{\IR}\left[\frac{d\IR^{\alpha,K}}{d\IR}\Big{|}\cF_t\right]
\end{eqnarray}


Note that $\IR^{\alpha,K} \in\cP$ for $\alpha \in (0,1]$, because for $s_1<s_2$:
\begin{eqnarray}
\IE^{\IR^{\alpha,K}}[X_{s_2}|\cF_{s_1}]&=&\frac{1}{{\frac{d\IR^{\alpha,K}}{d\IR}}|_{\cF_{s_1}}} \IE^{\IR}\left[\frac{d\IR^{\alpha,K}}{d\IR}X_{s_2}\Big{|}\cF_{s_1}\right]\nonumber\\
&=&\frac{1}{{\frac{d\IR^{\alpha,K}}{d\IR}}|_{\cF_{s_1}}}\IE^{\IR}\left[\IE^{\IR}\left[\frac{d\IR^{\alpha,K}}{d\IR}X_{s_2}\Big{|}\cF^X_T \vee \cF_{s_1} \right]\Big{|}\cF_{s_1}\right]\nonumber\\
&=&\IE^{\IR}\left[X_{s_2}\frac{\IE^{\IR}\left[\frac{d\IR^{\alpha,K}}{d\IR}\Big{|}\cF^X_T \vee \cF_{s_1}\right]}{{\frac{d\IR^{\alpha,K}}{d\IR}}|_{\cF_{s_1}}}\Big{|}\cF_{s_1}\right]\nonumber\\
&=&\IE^{\IR}\left[X_{s_2}|\cF_{s_1}\right]=X_{s_1}\nonumber
\end{eqnarray}
For every random variable $K \in \cK^M$ we have an inequality:
\begin{eqnarray}
\IE^{\IR}\overline{M}&\geq& \sup_{\alpha \in (0,1]}\IE^{\IR^{\alpha,K}}M \geq \IE^{\IR^{0,K}}M=\IE^{\IR}\left[ M\prod_{i=1}^{n} A^{\alpha,K}_i\right]\nonumber\\
&=&\IE^{\IR}\left[ M\prod_{i=1}^{n-1} A^{\alpha,K}_i  \frac{1_{\{M\geq K\}}}{\IR(M\geq K|\cF^X_{t_i}\vee \cF_{t_{i-1}} )} \right] \nonumber\\
&\geq&\IE^{\IR}\left[ K\prod_{i=1}^{n-1} A^{\alpha,K}_i  \frac{1_{\{M\geq K\}}}{\IR(M\geq K|\cF^X_{t_i}\vee \cF_{t_{i-1}} )} \right] \nonumber\\
&=&\IE^{\IR}\left[ K\prod_{i=1}^{n} A^{\alpha,K}_i\right]= \IE^{\IR}\left[K \IE^{\IR}\left[ \prod_{i=1}^{n} A^{\alpha,K}_i\Big{|}\cF_T^X\right]\right]=\IE^\IR[K]  \nonumber
\end{eqnarray}
but from the other side $\sup_{K \in \cK^M}\IE^{\IR}K=\IE^{\IR} \overline{M}$. This implies
\begin{equation}
\sup_{(\alpha,K) \in (0,1] \times\cK^M }\IE^{\IR^{\alpha,K}}M=\IE^{\IR} \overline{M}.
\end{equation}
\end{proof}

\section{Acknowledgements}
This author's research was supported by the Ministry of Science and
Higher Education grant NCN 2011/01/B/HS4/00982.

\bibliographystyle{plain}
{\small\bibliography{PhD}}
\end{document}